\documentclass[11pt]{amsart}

\usepackage[T1]{fontenc}
\usepackage{lmodern}
\usepackage{microtype}
\usepackage{amsmath,amssymb,mathtools}
\usepackage{enumitem}
\usepackage[hidelinks]{hyperref}
\hypersetup{
  pdftitle={Fine Difference Structure and Prime-Power Depth of Bent Partitions},
  pdfauthor={Zhaorui Wu},
  pdfsubject={Bent partitions, partitioned difference families, and prime-power depth}
}

\newtheorem{theorem}{Theorem}[section]
\newtheorem{corollary}[theorem]{Corollary}
\newtheorem{proposition}[theorem]{Proposition}
\newtheorem{lemma}[theorem]{Lemma}
\theoremstyle{definition}
\newtheorem{definition}[theorem]{Definition}
\theoremstyle{remark}

\newcommand{\F}{\mathbb{F}}
\newcommand{\1}{\mathbf{1}}

\title[Fine difference structure and prime-power depth]
{Fine Difference Structure and\\
Prime-Power Depth of Bent Partitions}

\author{Zhaorui Wu}
\address{University of Oxford}
\email{shil6872@ox.ac.uk}
\date{28 August 2026}


\begin{document}

\begin{abstract}
A $p$-ary bent partition of $\F_p^n$ is a partition into $K$ nonempty
cells such that every balanced assignment of its cells to $\F_p$ produces a
bent function.  It was asked whether every possible depth $K$ is a power of
$p$; for general $p$, previous affirmative results required regularity or
cell-symmetry hypotheses.  We prove the stronger unconditional statement
that, for every
nonzero $h$, exactly $p^n/K$ points remain in the same fine cell under
translation by $h$.  Thus the fine cells form a partitioned difference family
and the fine label map is zero-difference balanced.  Consequently
$K\mid p^n$, so $K=p^t$; nonempty cells further give $1\le t<n$.  In even
dimension, the classical cell-size theorem yields $K\mid p^{n/2}$.
Together with the known odd-dimensional ternary three-fibre parameter
restriction, this
gives the global bound $t\le\lfloor n/2\rfloor$.  The proof is an exact finite
average over balanced coarsenings.  The main counting identity and selected
consequences are formalized and kernel-checked in Lean 4.
\end{abstract}

\maketitle

\section{Introduction}

Let $p$ be prime and $V=\F_p^n$.  A bent partition in the sense of
Anbar and Meidl~\cite{AnbarMeidl2022} is a partition
\[
  V=A_1\sqcup\cdots\sqcup A_K
\]
with $p\mid K$ such that every assignment of the $K$ cells to the $p$ values
of $\F_p$, using every value equally often, produces a $p$-ary bent function.
The hypothesis is universal: it concerns every balanced coarsening of one
fixed fine partition.

All examples known in the founding paper had prime-power depth, and the
authors asked whether this must always hold~\cite[Section~6]{AnbarMeidl2022}.
The question remained open in the 2026 survey
\cite{AnbarKalayciMeidl2026}.  Wang, Wei, and Fu proved prime-power depth for
their class $\mathcal{WBP}$, in which the generated bent functions have a
common regularity type~\cite{WangWeiFu2026}; this includes every Boolean bent
partition.  A later result for the bent specialization of plateaued
partitions assumes cellwise negation symmetry in odd characteristic
\cite{WangWeiFuLiLi2026}.  The argument below removes both restrictions and
proves a stronger fine-partition statement.

Write $L:V\to I$ for the fine label map and fix $h\ne0$.  The quantity
\[
 D_h=\#\{x\in V:L(x+h)=L(x)\}
\]
is hidden from any one coarse function.  If $K=pm$, however, two distinct
fine labels collide under a uniformly random balanced coarsening with
probability $(m-1)/(K-1)$.  Every coarse derivative has exactly $p^{n-1}$
zeros.  Averaging therefore gives
\[
 D_h+(p^n-D_h)\frac{m-1}{K-1}=p^{n-1},
 \qquad\text{hence}\qquad KD_h=p^n.
\]
This identifies the fine cells as a partitioned difference family (PDF), or
equivalently the map $L$ as zero-difference balanced (ZDB), before any
prime-power structure on the label set is known.

The paper has three principal contributions.
\begin{enumerate}[label=(\roman*),leftmargin=*]
  \item Every bent partition is a PDF and its label map is ZDB, with exact
        nonzero-difference multiplicity $p^n/K$.
  \item Hence $K\mid p^n$ and $K=p^t$ without regularity, weak regularity, or
        cell-symmetry assumptions; moreover $1\le t<n$.
  \item Combining the new divisibility with the established
        dimension-and-parameter restriction gives
        $t\le\lfloor n/2\rfloor$, with the
        even-dimensional and odd ternary cases treated under their correct
        hypotheses.
\end{enumerate}

The collision coefficient used in the proof is classical in universal
hashing and resolvable designs~\cite{Stinson1994,LiuChen2013}, and PDFs and
ZDB maps are established objects~\cite{DingYin2005,Ding2008,
BurattiJungnickel2019,WangZhou2014}.  The contribution is the recovery
implication from the bent-partition axiom.  To the best of our knowledge, it
has not previously been stated in this unconditional form.

Bent partitions have also been studied through partial difference sets,
vectorial dual-bent functions, and LP-packings
\cite{AnbarKalayciMeidl2022PDS,WangFuWei2023PDS,
AlkanAnbarKalayciMeidl2024}.  Recent work shows both that these stronger
structures need not accompany every bent partition and that twisted variants
remain available in wider classes
\cite{AnbarFuKalayciMeidlWangWei2026,
AnbarKalayciMeidl2026Twisted}.  None of those structural correspondences
supplies the unconditional recovery argument used here.

\section{Definitions}

For $f:V\to\F_p$ and $h\in V$, write
\[
  \Delta_h f(x)=f(x+h)-f(x).
\]
We use the standard derivative consequence of $p$-ary bentness:
\[
  \#\{x\in V:\Delta_hf(x)=a\}=p^{n-1}
  \qquad(h\ne0,\ a\in\F_p).
\tag{2.1}\label{eq:bent-derivative}
\]
For clarity, this implication does not hide a regularity assumption.  With
$\zeta_p=e^{2\pi i/p}$, write
\[
 W_f(u)=\sum_{x\in V}\zeta_p^{f(x)-u\cdot x}.
\]
Walsh flatness and finite Fourier inversion (equivalently,
Wiener--Khinchin) give
\(
 \sum_x\zeta_p^{\Delta_hf(x)}=0
\)
for $h\ne0$.  If
$N_a=\#\{x:\Delta_hf(x)=a\}$, then
$\sum_aN_a\zeta_p^a=0$.  Define
\[
 Q(X)=\sum_{a=0}^{p-1}N_aX^a\in\mathbb Z[X].
\]
Then $Q(\zeta_p)=0$ and
$\deg Q\le p-1=\deg\Phi_p$.  Since $p$ is prime and
$\Phi_p(X)=1+X+\cdots+X^{p-1}$ is the minimal polynomial of $\zeta_p$,
we have $Q=c\Phi_p$ for some $c\in\mathbb Z$.  Thus all $N_a=c$; their sum
is $p^n$, proving~\eqref{eq:bent-derivative}.  This is the standard
perfect-nonlinearity characterization; see, for example,
\cite[Theorem~5 and Proposition~8]{CarletDing2004}.

Let $I$ be a $K$-element set, let $L:V\to I$ be surjective, and write
$A_i=L^{-1}(i)$.  Put $K=pm$.

\begin{definition}
The fibre partition of $L$ is a $p$-ary \emph{bent partition} of depth $K$
if, for every map $c:I\to\F_p$ satisfying
\[
  |c^{-1}(a)|=m\qquad(a\in\F_p),
\]
the coarsening $c\circ L$ is bent.
\end{definition}

This is the Anbar--Meidl definition.  It is distinct from a \emph{normal}
bent partition, which has a distinguished cell and a different counting law
\cite{AnbarMeidl2022}.

For a finite additive group $G$, a family of disjoint blocks
$\{A_i:i\in I\}$ partitioning $G$ is a \emph{partitioned difference family}
of index $\lambda$ if
\[
  \sum_{i\in I}|A_i\cap(A_i-h)|=\lambda
  \qquad(h\in G\setminus\{0\}).
\tag{2.2}\label{eq:pdf}
\]
The label map $L:G\to I$ is \emph{zero-difference balanced} with parameter
$\lambda$ if
\[
  \#\{x\in G:L(x+h)=L(x)\}=\lambda
  \qquad(h\ne0).
\tag{2.3}\label{eq:zdb}
\]
Equations~\eqref{eq:pdf} and~\eqref{eq:zdb} are the same condition, expressed
in the PDF and ZDB languages, respectively
\cite{Ding2008,BurattiJungnickel2019,WangZhou2014}.
After identifying $I$ with any group of order $K$, equation~\eqref{eq:zdb}
is the usual ZDB condition.  Its equality-only formulation is independent of
the chosen group structure on $I$.

\section{A balanced-fusion lemma}

The combinatorial step does not require a group or a Walsh transform.

\begin{lemma}[Balanced-fusion diagonal recovery]
\label{lem:recovery}
Let $X$ and $I$ be finite sets with $|X|=N$ and $|I|=K=bm$, where
$b>1$ and $m>0$.  Let $T:X\to X$ and $L:X\to I$.  For a balanced map
$c:I\to[b]$, set
\[
 Z_c(T)=\#\{x:c(L(Tx))=c(L(x))\},
 \qquad
 D_T=\#\{x:L(Tx)=L(x)\}.
\]
If $Z_c(T)=N/b$ for every balanced $c$, then $KD_T=N$.
\end{lemma}

\begin{proof}
Choose $c$ uniformly from the balanced maps.  Equal labels always collide.
For distinct $i,j\in I$, conditioning on $c(i)$ shows that
\[
  \Pr[c(i)=c(j)]=\frac{m-1}{K-1}.
\]
Pointwise averaging of $Z_c(T)$ gives
\[
 D_T+(N-D_T)\frac{m-1}{K-1}=\frac Nb.
\]
Using $K=bm$ and clearing denominators yields $KD_T=N$.
\end{proof}

The probability notation abbreviates an exact average over a finite family.
More generally, the same proof works for any probability distribution
supported on balanced maps for which every distinct label pair has the same
collision probability.

\section{The fine PDF and prime-power depth}

\begin{theorem}[Fine PDF/ZDB theorem]
\label{thm:pdf-zdb}
Let $L:V\to I$ define a $p$-ary bent partition of depth $K$.  Then, for
every $h\ne0$,
\[
  K\,\#\{x\in V:L(x+h)=L(x)\}=p^n,
  \qquad
  \#\{x\in V:L(x+h)=L(x)\}=\frac{p^n}{K}.
\tag{4.1}\label{eq:Dh-main}
\]
Consequently, the fine cells form a PDF of index $p^n/K$, and $L$ is ZDB
with the same parameter.
\end{theorem}

\begin{proof}
Write $K=pm$.  For every balanced $c:I\to\F_p$, the coarsening $c\circ L$
is bent.  Equation~\eqref{eq:bent-derivative}, with $a=0$, gives
\[
 \#\{x:c(L(x+h))=c(L(x))\}=p^{n-1}.
\]
Apply Lemma~\ref{lem:recovery} with $X=V$, $T(x)=x+h$, and $b=p$.
This gives $KD_h=p^n$, proving~\eqref{eq:Dh-main};
\eqref{eq:pdf}--\eqref{eq:zdb} give the stated interpretations.
\end{proof}

\begin{corollary}[Prime-power depth]
\label{cor:prime-power}
Every $p$-ary bent partition has depth $K=p^t$ for an integer
$t$ satisfying $1\le t\le n$.
\end{corollary}

\begin{proof}
Surjectivity onto $K\ge p\ge2$ labels excludes $n=0$.  Hence $V$ contains a
nonzero $h$, and the integer identity in Theorem~\ref{thm:pdf-zdb} gives
$K\mid p^n$.  Primality of $p$ gives $K=p^t$, while the defining condition
$p\mid K$ gives $t\ge1$.
\end{proof}

This resolves the prime-power-depth question without either the common
regularity hypothesis of Wang--Wei--Fu~\cite{WangWeiFu2026} or the
cell-symmetry hypothesis used in~\cite{WangWeiFuLiLi2026}.  It also yields
the finer PDF/ZDB property rather than divisibility alone.

\begin{corollary}[Strict depth]
\label{cor:strict}
In Corollary~\ref{cor:prime-power}, $1\le t<n$; equivalently, $K<p^n$.
\end{corollary}

\begin{proof}
Surjectivity onto $K\ge p>1$ labels excludes $n=0$.  Choose $h\ne0$.
Equation~\eqref{eq:Dh-main} is positive, so some two distinct points have the
same label.  Thus the surjective map $L$ is not injective, whence
$K<|V|=p^n$.
\end{proof}

The dimensional refinement requires a case distinction that is absent from
the original even-dimensional cell-size statement.

\begin{corollary}[Dimension-sensitive bound]
\label{cor:dimension}
Let $K=p^t$ be the depth of a bent partition with nonempty cells.
Then $1\le t\le\lfloor n/2\rfloor$.  More precisely:
\begin{enumerate}[label=(\alph*),leftmargin=*]
  \item if $n$ is even, then $K\mid p^{n/2}$;
  \item if $n$ is odd, then the only possible case is $(p,K)=(3,3)$.
\end{enumerate}
\end{corollary}

\begin{proof}
Suppose first that $n$ is even and put $q=p^{n/2}$.  After relabelling the
cells, the classical cell-size theorem gives the common size of
$A_2,\ldots,A_K$ as $q(q-1)/K$ or $q(q+1)/K$, with a common sign
\cite[Theorem~3]{AnbarMeidl2022}.  Hence $p^t\mid q(q\pm1)$.  Since
$q\pm1$ is coprime to $p$, it follows that $p^t\mid q$, proving (a) and
$t\le\lfloor n/2\rfloor$.

For odd $n$, the cited dimension-and-parameter restriction gives the sole
possibility $(p,K)=(3,3)$; equivalently, the partition is the three-fibre
partition of a ternary bent function
\cite[Remark~6]{AlkanAnbarKalayciMeidl2024}
\cite[Remark~2]{AnbarKalayciMeidl2026}.  Thus $t=1$.
The nonempty-cell convention, together with Corollary~\ref{cor:strict},
excludes $n=1$; hence $n\ge3$ and $1\le\lfloor n/2\rfloor$.
\end{proof}

The divisibility $K\mid p^{n/2}$ is asserted only in the even-dimensional
case.  In odd dimension $p^{n/2}$ is not an integer, and the ternary exception
must not be passed through the even-dimensional cell-size formula.

\section{Consequences and limits}

Put $v=p^n$, $\mu=v/K$, and $s_i=|A_i|$.  The fine PDF identity immediately
fixes the second moment of the cell sizes.

\begin{proposition}[Cell-size moment]
\label{prop:sphere}
For every bent partition,
\[
  \sum_{i\in I}s_i^2=v+\mu(v-1),
  \qquad
  \sum_{i\in I}(s_i-\mu)^2=\mu(K-1).
\tag{5.1}\label{eq:sphere}
\]
\end{proposition}

\begin{proof}
Count ordered pairs $(x,y)$ with $L(x)=L(y)$ by $h=y-x$.  The zero
difference contributes $v$, and each of the other $v-1$ differences
contributes $\mu$, proving the first identity.  Expanding the second and using
$\sum_i s_i=v=K\mu$ proves the claim.
\end{proof}

This is the standard ZDB second-moment identity in centered integral form
\cite{Ding2008,WangZhou2014}.  It is a useful arithmetic sieve, not a separate
source of the depth theorem.

Since $K=p^t$, any bijection $I\to\F_p^t$ relabels $L$ as a surjective
vectorial bent map $F$: every nonzero linear component pulls back to a balanced
coarsening and is therefore bent.  This relabelling observation is already
recorded in the bent-partition literature
\cite[Remark~5]{AnbarKalayciMeidl2026}; the new point here is that its
prime-power hypothesis now holds automatically.  Classical perfect
nonlinearity then gives the full vector-derivative law
\[
 \#\{x:F(x+h)-F(x)=u\}=p^{n-t}
 \qquad(h\ne0,\ u\in\F_p^t),
\]
after the chosen relabelling; see~\cite{CarletDing2004}.  Thus the fine ZDB
identity is the zero-difference part of a classical stronger law once the
previously unstructured label set has been shown to have prime-power size.

For completeness, the main averaging argument also has a small certificate
when $(p,K)=(3,6)$.  Label the six cells by $0,\dots,5$ and consider
\[
C=\begin{pmatrix}
0&0&1&1&2&2\\
0&1&0&2&2&1\\
0&1&2&0&1&2\\
0&2&2&1&0&1\\
0&2&1&2&1&0
\end{pmatrix}.
\tag{5.2}\label{eq:five-rows}
\]
The equal-colour pairs in the five rows form a one-factorization of $K_6$.
Thus, for a fixed $h\ne0$, a same-label transition is counted five times and
a distinct-label transition once.  If all five coarsenings were bent, with
$N=3^n$, their zero-derivative counts would give
\[
  \frac{5N}{3}=N+4D_h,
  \qquad\text{so}\qquad 6D_h=N,
\]
which is impossible by parity.  Hence these five tests alone rule out depth
six.  Among unweighted families of balanced maps for which every unordered
distinct pair is monochromatic in the same positive integer number of rows,
at least five rows are necessary, because each row covers three of the
fifteen edges of $K_6$.  This does not rule out a different deduction from
four coarsenings outside this pair-uniform scheme.

No converse is established here.  Fine ZDB records only zero-difference
counts, and additional conditions on the other derivative values would be
needed to recover the universal bent-coarsening axiom.  Likewise,
prime-power depth is only a necessary condition for existence.  Normal bent
partitions have a distinguished cell and a two-level law, so
Theorem~\ref{thm:pdf-zdb} does not apply to them by simply adjoining that
cell.

The other derivative values define a larger inverse problem.  For $h\ne0$,
put $M_h(i,j)=\#\{x:L(x)=i,\ L(x+h)=j\}$.  For balanced
$c:I\to\F_p$ and $a\in\F_p$, define
\[
 \mathcal M_{c,a}(M)=\sum_{c(j)-c(i)=a}M_{ij}.
\tag{5.3}\label{eq:transition-measurement}
\]
Then $\mathcal M_{c,a}(M_h)$ counts the points where
$\Delta_h(c\circ L)=a$.

For a $K$-element set $I$, let
\[
 \mathcal Z=
 \left\{M\in\mathbb R^{I\times I}:
   \sum_jM_{ij}=0\ (i\in I),\quad
   \sum_iM_{ij}=0\ (j\in I)\right\}.
\tag{5.4}\label{eq:zero-margin-space}
\]

\begin{proposition}[Transition-measurement obstruction]
\label{prop:transition-kernel}
Let $K=pm\ge3$.  Then
\[
 \bigcap_{\substack{c:I\to\F_p\\ c\ \mathrm{balanced}}}
 \ \bigcap_{a\in\F_p}
 \ker\!\left(\mathcal M_{c,a}\big|_{\mathcal Z}\right)
\]
contains a $(K-1)$-dimensional symmetric subspace of trace-zero matrices.
\end{proposition}

\begin{proof}
For $d=(d_i)_{i\in I}$ with $\sum_i d_i=0$, define
\[
 H(d)_{ii}=(K-2)d_i,
 \qquad
 H(d)_{ij}=-(d_i+d_j)\quad(i\ne j).
\tag{5.5}\label{eq:transition-kernel}
\]
The assignment $d\mapsto H(d)$ is linear, every $H(d)$ is symmetric, and a
direct sum gives zero row and column sums.  If $a\ne0$,
each label has $m$ possible successors and $m$ possible predecessors in
\eqref{eq:transition-measurement}, and no diagonal pair occurs.  Hence
\[
 \mathcal M_{c,a}(H(d))
 =-m\sum_i d_i-m\sum_j d_j=0.
\]
For $a=0$, use
$H(d)_{ij}=-(d_i+d_j)+K d_i\1_{i=j}$ to obtain
\[
 \mathcal M_{c,0}(H(d))
 =-2m\sum_i d_i+K\sum_i d_i=0.
\]
Thus $H(d)$ lies in the displayed common kernel.  Since $K\ge3$, its diagonal recovers
$d$, so $d\mapsto H(d)$ is injective on the $(K-1)$-dimensional hyperplane
\(\sum_i d_i=0\).  Finally,
\(\operatorname{tr}H(d)=(K-2)\sum_i d_i=0\).
\end{proof}

The proposition shows that the coarse derivative measurements are not
injective on the natural zero-margin perturbation space, while leaving the
trace recovered in Theorem~\ref{thm:pdf-zdb} unaffected.  It does not
assert that every $H(d)$ is realizable; finer rigidity requires additional
nonlinear or cross-shift information.

\section{Lean formalization}
\label{sec:lean}

The balanced-fusion lemma, the fine same-label count underlying the PDF/ZDB
theorem and its principal arithmetic consequences, the cell-size moment, the
five-test certificate, and the transition-measurement obstruction were formalized and
kernel-checked in Lean 4.32.0 with
mathlib 4.32.0~\cite{Lean4,Mathlib2020}.  The even-dimensional arithmetic
deduction accepts the classical cell-size divisibility as an explicit input;
the formalization does not claim to reprove that external analytic theorem or
the odd-dimensional classification.  Replay instructions, source-level scope
notes, axiom audits, and cryptographic hashes accompany the manuscript.

\section{Conclusion}

The universal balanced-coarsening axiom determines the hidden same-cell
diagonal exactly.  For every nonzero translation,
\[
  D_h=p^n/K.
\]
Therefore every bent partition is a PDF, its label map is ZDB, and its depth
divides $p^n$.  The formerly conditional prime-power conclusion follows
without regularity or symmetry assumptions.  The proof uses only derivative
balance and an exact finite collision average.

Although the nonzero derivative values contain more information than the
zero count used in the proof, Proposition~\ref{prop:transition-kernel} shows
that the associated measurement operator is non-injective on the ambient real
zero-margin perturbation space.  This does not assert two realizable
bent-partition transition matrices with identical data.  The trace remains
the statistic recovered by the universal axiom.

\bibliographystyle{amsplain}
\bibliography{references}

\end{document}